%% file: main.tex
\documentclass[11pt]{article}

\usepackage[utf8]{inputenc}

\usepackage{amsmath,amsfonts,amsthm,amssymb,color}
\usepackage[usenames,dvipsnames,svgnames,table]{xcolor}
\definecolor{darkgreen}{rgb}{0.0,0,0.9}
\usepackage{mathtools}
\usepackage{authblk}
\usepackage{fullpage}
\usepackage{parskip}
\usepackage{comment}
\usepackage{tikz}
\usepackage{bbm}
\usepackage{bm}
\usepackage{dsfont}
\usepackage[sc]{mathpazo}
\usepackage[basic]{complexity}
\usepackage{algorithm2e}
\usepackage[colorlinks=true,
citecolor=OliveGreen,linkcolor=BrickRed,urlcolor=BrickRed,pdfstartview=FitH]{hyperref}
\usepackage[capitalize,nameinlink]{cleveref}
\usepackage{tcolorbox}

\newtcolorbox{wbox}
{
	colback  = white,
}

\SetKwInOut{Input}{Input}
\SetKwInOut{Output}{Output}
\SetKwFunction{Uncross}{\textsc{Uncross}}
\SetKwFunction{MergeUncross}{\textsc{MergeUncross}}
\SetKwFunction{PerfectMatching}{\textsc{PerfectMatching}}
\SetKwBlock{InParallel}{in parallel do}{end}
\SetKwFor{ParallelFor}{for}{in parallel do}{end}
\newcommand*{\suppress}[1]{}

\newcommand*{\cM}{\mathcal{M}}

\newcommand*{\cR}{\mathcal{R}}
\newcommand*{\LLC}{\mathcal{L}}

\makeatletter
\def\thm@space@setup{%
	\thm@preskip= 10pt
	\thm@postskip=\thm@preskip % or whatever, if you don't want them to be equal
}
\makeatother

\makeatletter
\renewcommand{\paragraph}{%
	\@startsection{paragraph}{4}%
	{\z@}{5pt}{-1em}%
	{\normalfont\normalsize\bfseries}%
}
\makeatother

\newtheorem{theorem}{Theorem}
\newtheorem{lemma}{Lemma}
\newtheorem{corollary}{Corollary}
\newtheorem{definition}{Definition}
\newtheorem{problem}{Problem}

\newtheorem{example}{Example}
\newtheorem{observation}{Observation}
\newtheorem{proposition}{Proposition}

\theoremstyle{definition}

\newenvironment{fminipage}%
{\begin{Sbox}\begin{minipage}}%
		{\end{minipage}\end{Sbox}\fbox{\TheSbox}}

\def\union{\cup}

\newcommand*{\PSch}{\mbox{\rm{Preferred-Schools}}}

\newcommand*{\PSt}{\mbox{\rm{Preferring-Students}}}

\newcommand*{\Barr}{\mbox{\rm{Barrier}}}

\newcommand*{\FSch}{\mbox{\rm{Feasible-Schools}}}

\def\lla{\leftarrow}

\newcommand*{\BStP}{\mbox{\rm{BS-Preferring}}}

\newcommand*{\LPSt}{\mbox{\rm{LPS-Assigned}}}

\newcommand{\PgSch}{\text{Schools-FBPairs}}

\title{Stability-Preserving, Time-Efficient Mechanisms for \\
School Choice in Two Rounds\footnote{Supported in part by NSF grant CCF-1815901.}}

\author[1]{Karthik Gajulapalli}

\author[2]{James Liu\footnote{Part of this work was done while the author was a graduate student at the University of California, Irvine.}}
\author[3]{Tung Mai\footnote{Part of this work was done while the author was a postdoctoral fellow at the University of California, Irvine.}}
\author[1]{Vijay V.~Vazirani}

\affil[1]{University of California, Irvine}
\affil[2]{K-Sky Limited}
\affil[3]{Adobe Research}
\date{}

\begin{document}
	\maketitle
	
	\input{introduction}

	\input{overview}

	\input{prelim}

    \input{alg}

	%\input{s3}
	
	\input{enumerate}

	\input{B1}

	\input{B2}

	%\input{unit_case}

    %\input{S_3_4}
    
    \input{incentive_compatibility}

    \input{hardness}

	\input{discussion}

	\input{ack}

	\bibliographystyle{alpha}
	\bibliography{refs}
\end{document}

%% file: introduction.tex
\begin{abstract}
We address the following dynamic version of the school choice question: a city, named City, admits students in two temporally-separated rounds, denoted $\cR_1$ and $\cR_2$. In round $\cR_1$, the capacity of each school is fixed and mechanism $\cM_1$ finds a student optimal stable matching. In round $\cR_2$, certain parameters change, e.g., new students move into the City or the City is happy to allocate extra seats to specific schools. We study a number of Settings of this kind and give polynomial time algorithms for obtaining a stable matching for the new situations. 

It is well established that switching the school of a student midway, unsynchronized with her classmates, can cause traumatic effects. This fact guides us to two types of results, the first simply disallows any re-allocations in round $\cR_2$, and the second asks for a stable matching that minimizes the number of re-allocations. For the latter, we prove that the stable matchings which minimize the number of re-allocations form a sublattice of the lattice of stable matchings. Observations about incentive compatibility are woven into these results. We also give a third type of results, namely proofs of NP-hardness for a mechanism for round $\cR_2$ under certain settings. 
\end{abstract}

\section{Introduction}
\label{sec:intro}
School choice is among the most consequential events in a child's upbringing, whether it is admission to elementary, middle or high school, and hence has been accorded its due importance not only in the education literature but also in game theory and economics. In order to deal with the flaws in the practices of the day, the seminal paper of Abdulkadiroglu and Sonmez \cite{Atila-Sonmez} formulated this as a mechanism design problem. This approach has been enormously successful, especially in large cities involving the admission of tens of thousands of students into hundreds of schools, e.g., see \cite{NYC, New-Orleans, abdulkadiroglu2013matching, Pathak}, and today occupies a key place in the area of market design in economics, e.g., see \cite{Nobel, roth2008, Roth, Simons}. 

Once the basic game-theoretic issues in school choice were adequately addressed, researchers turned attention to the next level of questions. In this vein, in a recent paper, Feigenbaum et. al. \cite{Feigenbaum} remarked, ``However, most models considered in this literature are essentially static. Incorporating dynamic considerations in designing assignment mechanisms ... is an important aspect that has only recently started to be addressed.'' 

Our paper deals with precisely this. We define several settings for school choice in which an instance is made available in the first round $\cR_1$ and at a later time, in the second round $\cR_2$, some of the parameters change. Each setting asks for a pair of mechanisms, $(\cM_1, \ \cM_2)$ for finding matchings of students to schools in these two rounds. All our settings insist that the matchings found in both rounds are stable. It will be convenient to classify our results into three types. In Type A and B, both mechanisms $\cM_1$ and $\cM_2$ are required to run in polynomial time.
\begin{enumerate}
	\item {\bf Type A:} Mechanism $\cM_2$ is disallowed to reassign the school of any student matched by $\cM_1$. We present two settings, {\bf A1} and {\bf A2}. 
	\item {\bf Type B:} Mechanism $\cM_2$ is allowed to reassign the school of students matched by $\cM_1$; however, it needs to (provably) minimize the number of such reassignments. We present two settings, {\bf B1} and {\bf B2}.
	\item {\bf Type C:} These are NP-hardness results -- of mechanism $\cM_2$ for four problems and of a fifth problem, which involves only one round.
\end{enumerate}

\subsection{Our model and its justification}
\label{sec.model}

Our solutions to Type A and B results will strictly adhere to the following tenets; we justify them below.

\begin{enumerate}
	\item {\bf Tenet 1:}  All matchings produced by our mechanisms need to be stable.
	\item {\bf Tenet 2:}  In Type A results, mechanism $\cM_2$ is disallowed to reassign the school of any student matched by $\cM_1$, and in Type B results, $\cM_2$ must provably minimize the number of such reassignments.
	\item {\bf Tenet 3:} We want all our mechanisms to run in polynomial time.
\end{enumerate}

The use of the classic Gale-Shapley \cite{GaleS} Deferred Acceptance Algorithm has emerged as a method of choice in the literature. Our mechanisms also use this algorithm. Stability comes with key advantages: First, no student and school, who are not matched to each other, will have the incentive to go outside the mechanism to strike a deal. Second, it eliminates {\em justified envy}, i.e., the following situation cannot arise: there is a student $s_i$ who prefers another student $s_j$'s school assignment, say $h_k$, while being fully aware that $h_k$ preferred her to $s_j$.

Switching the school of a student midway, unsynchronized with her classmates -- such as when the entire class moves from elementary to middle or from middle to high school -- is well-known to cause traumatic effects, e.g., see \cite{switching}. It is for these reasons that in Type A results, mechanism $\cM_2$ is disallowed to reassign the school of any student matched by $\cM_1$ and in Type B results, $\cM_2$ must provably minimize the number of such reassignments. For Type A results, we say that $\cM_2$ {\em extends} $M$ to a stable matching $M'$. For Type B results, we say that $\cM_2 $ computes a {\em minimum stable re-allocation} $M'$ of $M$.

The strongest notion of incentive compatibility for a mechanism is {\em dominant strategy incentive compatible (DSIC)}, for students. This entails that regardless of the preferences reported by other students, a student can do no better than report her true preference list, i.e., truth-telling is a dominant strategy for all students. This immediately simplifies the task of students and their parents, since they don't need to waste any effort trying to game the system. Furthermore, if students are forced to adjust their choices in an attempt to gain a better matching, the mechanism, dealing with choices reported to it, may be forced to make matches that are suboptimal for students as well as schools. 

Gale and Shapley \cite{GaleS} proved that if the Differed Acceptance Algorithm is run with students proposing, it will yield a {\em student-optimal matching}, i.e., each student will get the best possible school, according to her preference list, among all stable matchings. However, this matching may be extremely unfavorable to an individual student --  she may be matched to a school which is very low on her preference list, giving her incentive to cheat, i.e., provide a false preference list, in order to get a better matching. Almost two decades after the Gale-Shapley result, Dubins and Freedman \cite{Dubins1981machiavelli} proved, via a highly non-trivial analysis, that this algorithm is DSIC for students. This ground-breaking result opened up the Gale-Shapley algorithm to a host of highly consequential applications, including school choice.

It is worth noting that Dubins and Freedman proved incentive compatibility for {\em precisely} the algorithm given by Gale and Shapley twenty years earlier. The latter did not even mention incentive compatibility in their paper and it is safe to assume that they did not design their algorithm around this fact. Since polynomial time solvability exploits the underlying structure of the problem in such a profound manner, it is reasonable to assume that it bestows the resulting algorithm with desirable properties\footnote{This fact is related to the notion of ``algorithmic way of thinking'' or the ``computational lens'', which is believed to yield progress of a {\em fundamental nature} in the other sciences.}, such as DSIC. 

In all of our results of Type A and B, mechanism $\cM_1$ finds a student-optimal stable matching using the Gale-Shapley Differed Acceptance Algorithm and is therefore DSIC for students. For Setting B2 we provide a mechanism for round $\cR_2$ that is DSIC. However, our mechanisms for round $\cR_2$ for the remaining three settings do not achieve this. Our main open problem is to fix this. For completeness, and in order to motivate this open problem, we provide counter-examples for each of these setting in Section $\ref{sec:incentives}$.

It is well known that the set of Stable Matchings of a given instance forms a finite distributive lattice \cite{GusfieldI}. By orienting the underlying partial order of this lattice appropriately, the student-optimal stable matching can be made the top element of this lattice and the school optimal matching the bottom element. For both Settings of Type $B$, we show that the set of minimum stable re-allocations form a sublattice of this lattice. We provide  polynomial-time mechanisms for computing the top and bottom elements of this sublattice. For Setting B1, we show that the top of the sublattice is also the top of the whole lattice, i.e., it is the student-optimal stable matching; this is crucial for showing DSIC for B1.

\subsubsection{Type A and B settings}
\label{sec.settings}

In all four settings, we give a pair of mechanisms $(\cM_1, \cM_2)$ which run in polynomial time. The four settings involve the admission of students of a city, named City, into schools; the preference lists of both students and schools are provided to the mechanisms.  $\cM_1$ computes a student-optimal stable matching, $M$, over all the participants in $\cR_1$. In $\cR_2$ some of the parameters over which $M$ was defined are updated. $\cM_2$ then modifies the matching $M$ to produce a new matching $M'$ that is stable over the new parameters defined in $\cR_2$. 

For Settings of Type A, in round $\cR_1$, the capacity of each school is fixed but in round $\cR_2$, the City is happy to allocate extra seats to specific schools per the recommendation of mechanism $\cM_2$, which in turn has to meet specified requirements imposed by the City.  Let $L$ be the set of {\em left-over students}, those who could not be admitted in round $\cR_1$. 

In round $\cR_2$ of Setting A1, the problem is to maximize the number of students admitted from $L$, by extending $M$ in a stability-preserving manner. In Setting A2, a set $N$ of {\em new students} also arrive from other cities and their preference lists are revealed to $\cM_2$. The requirement now is to admit as few students as possible from $N$ and subject to that, as many as possible from $L$, again in a stability-preserving manner. Next, we consider a slightly different problem within Setting A2, namely find the largest subset of $(N \cup L)$ that can be matched in a stability-preserving manner. We give an efficient mechanism for this as well. Finally, we give a procedure that outputs all possible stability-preserving extensions of a given stable matching (which may be exponentially many) with polynomial delay.

For Settings of Type B, the capacity of each school is fixed in $\cR_1$, but in $\cR_2$ the City has to deal with the arrival of new students and new schools. This could lead the matching found by $\cM_1$ to no longer be stable. 

In round $\cR_2$ of Setting B1, a set $N$ of new students arrive and their preference lists are revealed to $\cM_2$. The capacity of schools remain unchanged and the problem is to find a matching, $M'$ that is stable under the arrival of new students which minimizes the number of students who are assigned to a different school in $M'$. In Setting B2, a set $H'$ of new schools arrive and the City allows the capacities of the original schools to increase. The preference lists of the students are updated to reflect these new schools, we again require that $\cM_2$ compute a new stable matching, $M'$ over the updated preference lists that minimizes the number of students who get matched to a different school in $M'$.

 \subsection{Related work}
\label{sec.related}

Besides the references pointed out above on school choice, in this section, we will concentrate on 
recent work on dynamic matching markets, especially those pertaining to school choice. Feigenbaum et. al. \cite{Feigenbaum} study the following issue that arises in NYC public high schools, which admits over 80,000 students annually: after the initial centralized allocation, about 10\% of the students choose not attend the school allocated to them, instead going to private or charter schools. To deal with this, \cite{Feigenbaum} give a two-round solution which maintains truthfulness and efficiency and minimizes the movement of students between schools.

An interesting phenomena that has been observed in matching markets is {\em unraveling}, under which matches are made early to beat the competition, even though it leads to inefficiencies due to unavailability of full information. A classic case, indeed one that motivated the formation of centralized clearing houses, is that of the market for medical interns in which contracts for interns were signed two years before the future interns would even graduate \cite{stability}. A theoretical explanation of this phenomena was recently provided by \cite{echenique2016strategic}. 

% We note that the phenomena we are studying in Setting III can be viewed as anti-unraveling: some students are able to game the system by making the match {\em late}. Clearly, this aspect deserves more work. We also note that this phenomenon is by no means rare, e.g., it occurred in the Pasadena School District and the authorities were made specific recommendations by economists from Caltech to counter it \cite{Fed}. 

\cite{kadam2018multiperiod} point out that stable pairings may not necessarily last forever, e.g., a student may switch from private to public school or a married couple may divorce. They study  dynamic, multi-period, bilateral matching markets and they define and identify sufficient conditions for the existence of a dynamically stable matching.

\cite{doval} develops a notion of stability that applies in markets where matching opportunities arrive over time, much like the seats in our work. One of the things shown in this paper is that agents' incentive to wait for better matching opportunities can make achieving stability very difficult. Indeed, the notion of dynamic stability given in this paper is a necessary condition which a matching must satisfy in order that agents do not to find it profitable to game a mechanism by showing up in later rounds. 

A number of recent papers \cite{westkamp,dogan,andersson,dur,haeringer} consider the consequences of having a mechanism that repeats the Gale-Shapley Deferred Acceptance algorithm multiple times, similar to our work. Note that Deferred Acceptance is not {\em consistent} in that if one runs it, then removes some agents and their assignments, and runs it again on the remaining agents, one does not obtain the same assignment restricted to the left-over agents. In these papers, the authors show that there is room for manipulation by submitting empty lists in the first round. However, unlike our model in which changes are introduced in round $\cR_2$, in all these papers, there is nothing that motivates running Deferred Acceptance twice, namely no arrivals of new students, no change in capacities, no changes in preferences, etc. 

%Thus, the issues in these papers would be solved by having the second algorithm extend the outcome of the first in a stable way, like you do in your paper.

%%%%%%%%%%%%%%%%%%%%%%%%%%%%%%%%%%%%%%%%%

%% file: overview.tex
\subsection{Overview of structural and algorithmic ideas}
\label{sec.overview}

The main idea for obtaining a stability-preserving mechanism in round $\cR_2$ for Settings $A1$ and $A2$ lies in the notion of a {\em barrier} which ensures that students admitted in $\cR_2$ do not form blocking pairs. A crucial issue is to place barriers optimally to ensure that the number of students admitted is optimized (minimized or maximized) appropriately. 

The algorithm for enumerating stable extensions of a stable matching, given in Section~\ref{sec.enumerate}, relies heavily on the fundamental 
structural property of stable matchings given in Lemma~\ref{lem:blocking}. 
Enumerated matchings are extended by only one student in an iteration.
At each step, the algorithm finds all such feasible extensions by one student 
in a way such that 
there must be at least one feasible assignment, for any student, at each step. 
This assurance is crucial in guaranteeing that the delay between any two enumerated matchings is polynomial. 

For Settings B1 and B2, the mechanism proceeds by iteratively resolving blocking pairs. 
% The core of our analysis lies in looking at the special case where all schools have unit capacity. We then extend this to the more general case of schools having an arbitrary capacity function. 
Structurally, we  show that the set of all minimum stable re-allocations forms a sublattice of the stable matching lattice. The core of this analysis relies on the fact that the set of students who are assigned to a different school in round $\cR_2$ cannot be matched to their original school in any minimum stable re-allocation. This lets us divide the set of students into two groups, students matched to the same school (fixed students) in all minimum stable re-allocations, and students matched to different schools (moving students). We then construct a smaller stable matching instance, $I$, over the set of moving students. By appropriately placing barriers for each student and school in $I$ we can ensure that the union of any stable matching in $I$ and the matching restricted to the fixed students will also be stable. This stable matching is a minimum stable re-allocation and defines a bijection between the set of minimum stable re-allocations and set of stable matchings in $I$, we exploit the lattice structure of the latter.

%% file: prelim.tex
\section{Preliminaries}

\subsection{The stable matching problem for school choice}
\label{sec:stable}

The stable matching problem takes as input a set $H = \{h_1, h_2, \ldots , h_m\}$ of $m$ public schools and a set $S =\{s_1, s_2, \ldots , s_n\}$ of $n$ students who are seeking admission to the schools. Each school $h_j \in H$ has an integer-valued {\em capacity}, $c(j)$, stating the maximum number of students that can be assigned to it. If $h_j$ is assigned $c(j)$ students, we will say that $h_j$ is {\em filled}, and otherwise it is {\em under-filled}. 

Each student $s_i \in S$ has a strict and complete preference list, $l(s_i)$, over $H \cup \{\emptyset\}$. If $s_i$ prefers $\emptyset$ to $h_j$, then she prefers remaining unassigned rather being assigned to school $h_j$. We will assume that the list $l(s_i)$ is ordered by decreasing preferences. Therefore, if $s_i$ prefers $h_j$ to $h_k$, 
%denoted by $h_j \geq_{s_i} h_k$, 
we can equivalently say that $h_j$ appears {\em before} $h_k$ or $h_k$ appears {\em after} $h_j$ on $s_i$'s preference list. Clearly, the order among the schools occurring after $\emptyset$ on $s_i$'s list is immaterial, since $s_i$ prefers remaining unassigned rather than being assigned to any one of them. Similarly, each school $h_j \in H$ has a strict and complete preference list, $l(h_j)$, over $S \cup \{\emptyset\}$. Once again, for each student $s_i$ occurring after $\emptyset$, $h_j$ prefers remaining under-filled rather than admitting $s_i$, and the order among these students is of no consequence.

Given a set of schools, $H' \subseteq H$, by the {\em best school for $s_i$ in $H'$} we mean the school that $s_i$ prefers the most among the schools in $H'$. Similarly, given a set of students, $S' \subseteq S$, by the {\em best student for $h_j$ in $S'$} we mean the student whom $h_j$ prefers the most among the students in $S'$.

A {\em matching} $M$ is a function, $M: S \rightarrow H \cup \{\emptyset\}$ such that if $M(s_i) = h_j$ then it must be the case that $s_i$ prefers $h_j$ to $\emptyset$ and $h_j$ prefers $s_i$ to $\emptyset$; if so, we say that student $s_i$ is {\em assigned to} school $h_j$. If $M(s_i) = \emptyset$, then $s_i$ is not assigned to any school. The matching $M$ also has to ensure that the number of students assigned to each school $h_j$ is at most $c(j)$. 

For a matching $M$, a student-school pair $(s_i, h_j)$ is said to be a \emph{blocking pair} if $s_i$ is not assigned to $h_j$, $s_i$ prefers $h_j$ to $M(s_i)$ and one of the following conditions holds:
\begin{enumerate}
	\item $h_j$ prefers $s_i$ to one of the students assigned to $h_j$, or
	\item $h_j$ is under-filled and $h_j$ prefers $s_i$ to $\emptyset$.
\end{enumerate}
The blocking pair is said to be {\em type 1} ({\em type 2})  if the first (second) condition holds.
A matching $M$ is said to be \emph{stable} if there is no blocking pair for it.

% Most of the mechanisms presented in this paper are {\em dominant-strategy incentive-compatibile},  for the students, i.e., truth-telling is a weakly-dominant strategy for students: they cannot gain by being untruthful, regardless of what the others do. We will often shorten this term to {\em DSIC mechanism}.
\begin{theorem} \textbf{Rural Hospitals Theorem} \cite{rural_hospital}
\begin{enumerate}
	\item Over all the stable matchings of the given instance: the set of matched students is the same and the number of students matched to each school is also the same.

	\item Assume that school $h$ is not matched to capacity in a stable matching. Then, the set of students matched to $h$ is the same over all stable matchings.

\end{enumerate}

\end{theorem}

\subsection{The Stable Matching Lattice}

Let $S$, be a finite set of elements and $\succeq$ be a partial order over the elements in $S$. Given two elements $a,b \in S$, we define an upperbound of $a,b$ as an element $u \in S$ such that $u \succeq a$, and $u \succeq b$. Similarly we define a lowerbound of $a,b$ as an element $l \in S$ such that $a \succeq l$, and $b \succeq l$. 
An element $u \in S$ is considered the least upperbound of $(a,b)$:\
	\begin{itemize}
		\item $u$ is an upperbound of $(a,b)$
		\item for all upperbounds $u' \in S$, $u' \succeq u$
	\end{itemize}

\noindent Greatest lowerbound can be defined similarly.

\begin{definition}  A Lattice $\LLC = (S, \succeq)$, is defined over a finite set $S$, and a partial order $\succeq$, if for every pair of elements $a,b \in S$, there exists a unique least upperbound and a unique greatest lowerbound. We call the least upperbound the join of $a$ and $b$ and denote it by $a \vee b$, and anagolously call the least lowerbound the meet of $a$ and $b$ and denote it by $a \wedge b$ 

\end{definition}

\begin{definition}
	Let $SM$ denote the set of stable matchings over given instance $(S, H, c)$, then for two stable matchings $M, M' \subseteq SM$, $M \succeq M'$, if and only if $ \ \forall$ $s_i \in S, \ s_i $ weakly prefers $M(s_i)$ to $M'(s_i)$ 

\end{definition}

\noindent Given two stable matchings $M$ and $M'$ consider two new maps $M_U$ and $M_L$, defined as follows:
\begin{itemize}
	\item$M_U(s_i) = max(M(s_i), M'(s_i))$
	\item$M_L(s_i) = min(M(s_i), M'(s_i))$
\end{itemize}
where max is the partner $s_i$ weakly prefers between $M$ and $M'$, and min is the complement of max.\\

\begin{theorem}(\cite{GusfieldI}).\label{sec:stable-matching-lattice} 
	The set of stable matchings $(SM,\succeq)$ characterizes a finite distributive lattice. Morever $M_L$, $M_U$ represent the meet and join of any two stable matchings in the lattice.\\
\end{theorem}

%%%%%%%%%%%%%%%%%%%%%%%%%

\section{Our Results for the Four Settings}
\label{sec:prob}

% As stated above, in this paper, we will study assignment of students to schools in two rounds, $\cR_1$ and $\cR_2$, which are temporally separated. 
%We now state the common aspects of the four settings before describing them completely. 
In round $\cR_1$, the setup defined in Section \ref{sec:stable} prevails and $\cM_1$ simply computes the student-optimal stable matching respecting the capacity of each school, namely $c(j)$ for $h_j$. Let this matching be denoted by $M$, $S_M \subseteq S$ be the set of students assigned to schools by $M$ and $L = (S -S_M)$ be the set of {\em left-over students}. As shown in \cite{Machiavelli}, $\cM_1$ is DSIC for students. 

For Settings of Type A, in round $\cR_2$, 
%$\cM_2$ extends $M$ to $M'$, which is also required to be stable. By {\em extends} we mean that $\cM_2$ is not allowed to break a match created by $\cM_1$; it can only match additional student-school pairs. In round $\cR_2$, 
the City has decided to extend matching $M$ in a stable manner without any restrictions on extra capacity added to each school. Let us denote this matching by $M'$ and let $c'(j)$ be the number of students matched to school $h_j$, equivalently the {\em round $\cR_2$  capacity of school $h_j$}, for $h_j \in H$. Once we obtain the solution under this assumption, we will show how it can be modified in case the City can only add a restricted number of extra seats to each school.

% In this section we state the four settings studied; for each, we will have two mechanisms, $\cM_1$ and $\cM_2$.  In round $\cR_1$, mechanism $\cM_1$ finds a stable matching of students to schools, $M$. 

For Settings of Type B, in round $\cR_2$ a change is made to the sets of participants, which may cause $M$ to no longer be a valid or stable matching. $\cM_2$ then updates $M$ to $M'$ in order to ensure a stable matching. By allowing {\em updates}, we let some students in $M$ get unmatched in $M'$, or get matched to different schools.  The City would like to minimize the number of students who would have to change schools, or no longer be matched to a school, in going from $M$ to $M'$. We call $M'$ a \textbf{minimum stable re-allocation} of $M$. Formally, M' is a minimum stable re-allocation of $M$ if $M'$ is a stable matching over all participants and the number of students $s_i \in S_M$ where $M(s_i) \neq M'(s_i)$ is minimized.

% The students report their preference lists and the mechanisms operate on whatever is reported. We will assume that the schools' preference lists are truthfully reported and we will show that in each of the three settings, $(\cM_1, \cM_2)$ are DSIC for students, hence showing that the students gain nothing by misreporting their preference lists.  

% We now state the common aspects of the first two settings before describing them completely; the third setting is quite different. In both, in round $\cR_1$, the setup defined in Section \ref{sec:stable} prevails and $\cM_1$ simply computes the student-optimal stable matching respecting the capacity of each school, namely $c(j)$ for $h_j$. Let this matching be denoted by $M$, $S_M \subseteq S$ be the set of students assigned to schools by $M$ and $L = (S -S_M)$ be the set of {\em left-over students}. As shown in \cite{Machiavelli}, $\cM_1$ is DSIC for students. 

% In round $\cR_2$, the City has decided to extend matching $M$ in a stable manner without any restrictions on extra capacity added to each school. Let us denote this matching by $M'$ and let $c'(j)$ be the number of students matched to school $h_j$, equivalently the {\em round $\cR_2$  capacity of school $h_j$}, for $h_j \in H$. Once we obtain the solution under this assumption, we will show how it can be modified in case the City can only add a restricted number of extra seats to each school.

%%%%%%%%%%%%%%%%%%%%%%%

\begin{lemma}
	For some student $s_i$, let $M(s_i) = h_j$. Then for any student $s_k \in L$, $s_k$ appears after $s_i$ in $l(h_j)$.
\end{lemma}

\begin{proof}
	If $s_k$ were to appear before $s_i$ in $l(h_j)$, then $(s_k, h_j)$ will form a blocking pair for $M$, contradicting its stability.
\end{proof}

\subsection{Setting A1}
\label{sec:S1}

In this setting, in round $\cR_2$, the City wants to admit as many students from $L$ as possible in a stablity-preserving manner. We will call this problem $\textsc{$Max_L$}$. We will prove the following:

\begin{theorem}
\label{thm:S1}
	There is a polynomial time mechanism $\cM_2$ that extends matching $M$ to $M'$ so that $M'$ is stable w.r.t. students $S$ and schools $H$. Furthermore, $\cM_2$ yields the largest matching that can be obtained by a mechanism satisfying the stated conditions.
\end{theorem} 

Let $k$ be the maximum number of students that can be added from $L$, as per Theorem \ref{thm:S1}. Next, suppose that the City can only afford to add $k' < k$ extra seats. We show in Section \ref{sec:mS1} how this can be achieved while maintaining all the properties stated in Theorem \ref{thm:S1}.

\subsection{Setting A2}
\label{sec:S2}

In this setting, in round $\cR_2$, in addition to the leftover set $L$, a set $N$ of {\em new students} arrive from other cities and their preference lists are revealed to mechanism $\cM_2$. Additionally, the schools also update their preference lists to include the new students. In this setting, the City wants to give preference to students who were not matched in round $\cR_1$, i.e., $L$, over the new students, $N$. Thus it seeks the subset of $N$ that {\em must} be admitted to avoid blocking pairs and subject to that, maximize the subset of $L$ that can be added, again in a stability-preserving manner. We will call this problem $\textsc{$Min_N Max_L$}$. We will prove the following:

\begin{theorem}
\label{thm:S2}
	There is a polynomial time mechanism $\cM_2$ that accomplishes the following:
	\begin{enumerate}
		\item It finds smallest subset $N' \subseteq N$ with which the current matching $M$ needs to be extended in a stability-preserving manner. 
		\item 
		Subject to the previous extension, it finds the largest subset $L' \subseteq L$ with which the matching can be extended further in a stability-preserving manner. 
	\end{enumerate}
\end{theorem}

\subsection{Setting B1}
\label{sec:S4}

In this setting, a set $N$ of {\em new students} arrive from other cities in round $\cR_2$. The preference lists of schools are also updated to include students in $N$, though their relative preferences between students in $S\cup \{\emptyset\}$ are unchanged. Since the capacities of schools are fixed, the addition of new students could result in some students being displaced, possibly causing them to go to another school or become unmatched.
The City wants to find a stable matching over students $S \cup N$ and schools $H$ that minimizes the number of students who are re-allocated from their original school in $M$.

\begin{theorem}
\label{thm:S4}

There is a polynomial time mechanism $\cM_2$ that finds a minimum stable reallocation with respect to Round $\cR_1$ matching $M$, students $S \cup N$, and schools $H$.
\end{theorem}

\subsection{Setting B2}
\label{sec:S3}

In this setting, the City has some new schools $H'$ that have opened up in $\cR_2$. The preference lists of students are updated to include schools in $H'$, though their relative preferences between schools in $H \cup \{\emptyset\}$ are unchanged. The City also allows schools in $H$ to increase their capacity in round $\cR_2$.
The addition of new schools could cause students in $L$ to get matched in $\cR_2$, but it might also result in some students wanting to leave their current schools to go to a new school. This could lead to vacant seats being created in the original schools, causing some other students to leave their current schools. The City wants to find a stable matching over students $S$ and schools $H\cup H'$ that minimizes the number of students who are re-allocated from their original school in $M$.
\begin{theorem}
\label{thm:S3}
	There is a polynomial time mechanism $\cM_2$ that finds a minimum stable reallocation with respect to Round $\cR_1$ matching $M$, students $S$, and schools $H\cup H'$.
\end{theorem}

%% file: alg.tex
\section{Mechanisms for Type A Settings}
\label{sec.mech}

\subsection{Setting A1}
\label{sec:mS1}

We will first characterize situations under which a matching is not stable, i.e., admits a blocking pair. This characterization will be used for proving stability of matchings constructed in round $\cR_2$. For this purpose, assume that $M$ is an arbitrary matching, not necessarily stable nor related to the matching computed in round $\cR_1$. For each school $h_j \in H$, define the {\em least preferred student assigned to $h_j$}, denoted $\LPSt(h_j)$, to be the student whom $h_j$ prefers the least among the students that are assigned to $h_j$. 

Next, for each student $s_i \in S_M$, define the {\em set of schools preferred by $s_i$}, denoted  $\PSch (s_i)$ by $\{ h_j ~|~ s_i \ \mbox{prefers} \ h_j \ \mbox{to} \ M(s_i) \}$; note that $M(s_i) = \emptyset$ is allowed in this definition. Further, for each school $h_j \in H$, define the {\em set of students that prefer $h_j$ over the school they are assigned to}, denoted $\PSt (h_j)$ to be $\{ s_i ~|~  h_j \in \PSch (s_i)  \}$. Finally, define the {\em best student preferring $h_j$}, denoted $\BStP(h_j)$, to be the student whom $h_j$ prefers the best in the set $\PSt(h_j)$. If $\PSt (h_j) = \emptyset$ then we will define $\BStP (h_j) = \emptyset$; in particular, this happens if $h_j$ is under-filled.

\begin{lemma}. 
	\label{lem:blocking}
	W.r.t. matching $M$, there exists a blocking pair:
	\begin{enumerate}
		\item of type 1 iff there is a school $h_j$ s.t. $h_j$ prefers $\BStP(h_j)$ to $\LPSt(h_j)$.
		\item of type 2 iff there is a school $h_j$ that is under-filled and a student $s_i$ such that $s_i$ prefers $h_j$ to $M(s_i)$ and $h_j$ prefers $s_i$ to $\emptyset$.
	\end{enumerate}
\end{lemma}

\begin{proof}
\begin{enumerate}
	\item Suppose for some school $h_j$, $\BStP(h_j) = s_i$ and $\LPSt(h_j) = s_k$ and $h_j$ prefers $s_i$ to $s_k$. Then, $s_i$ prefers $h_j$ to $M(s_i)$ and $h_j$ prefers $s_i$ to $s_k$. Therefore, $(s_i, h_j)$ is a blocking pair of type 1. Next, assume that $(s_i, h_j)$ is a blocking pair of type 1. Then, it must be the case that $s_i$ prefers $h_j$ to $M(s_i)$ and $h_j$ prefers $s_i$ to $s_k$, for some student $s_k$ that is assigned to $h_j$. Clearly, $h_j$ weakly prefers $s_k$ to $\LPSt(h_j)$. Therefore $h_j$ prefers $\BStP(h_j)$ to $\LPSt(h_j)$.
	\item Both directions follow from the definition of blocking pair of type 2. 
\end{enumerate}
\end{proof}

The mechanism $\cM_2$ for round $\cR_2$ for $\textsc{$Max_L$}$ in Setting A1 is given in Figure \ref{alg:S1}. Step 1 simply ensures that the matching found by $\cM_2$ extends the round $\cR_1$ matching. Step 2 defines the Barrier for each school to be $\BStP(h_j)$; observe that this could be $\emptyset$. Step 3 determines the set $L' \subseteq L$ that can be assigned schools in a stability-preserving manner and Step 5 computes the school for each student in this subset.

\begin{figure}
	\begin{wbox}
		\textsc{$Max_L$}$(M,L)$: \\
		\textbf{Input:} Stable matching $M$ and set $L$.   \\
		\textbf{Output:} Stable, IC, $Max_L$ extension of $M$. \\ 
		\begin{enumerate}
			\item $\forall s_i \in S_M: M'(s_i) \leftarrow M(s_i)$ \\			
			\item $\forall h_j \in H: \ \Barr (h_j) \lla \ \BStP (h_j)$. \\
			\item $L' \lla \{ s_i \in L  ~|~  \exists h_j \ \ \mbox{s.t. $s_i$ appears before $\Barr (h_j)$ in $l(h_j)$}$, \\
			\hspace*{2.4cm} $\mbox{and $h_j$ appears before $\emptyset$ in $l(s_i)$} \}$. \\
			\item $\forall s_i \in L': \ \FSch(s_i) \lla \{ h_j ~|~  s_i \ \mbox{appears before} \ 
                \Barr(h_j) \ \mbox{in} \ l(h_j) \}$. \\
            \item $\forall s_i \in L': \ M'(s_i) \lla \ \mbox{Best school for} \ s_i \ \mbox{in} \ \FSch(s_i)$. \\
            \item $\forall s_i \in (L - L'): \ M'(s_i) \lla \emptyset$. \\
            
            \item Return $M'$. 
			\end{enumerate}

	\end{wbox}
	\caption{Mechanism for round $\cR_2$ for problem $\textsc{$Max_L$}$ in Setting A1}
	\label{alg:S1} 
\end{figure} 

\begin{proof} {\em of Theorem \ref{thm:S1}:} 
Suppose $\Barr(h_j) = s_i$ (or, $\emptyset$). Since all students assigned to $h_j$ from $L$ appear before $s_i$ (respectively, $\emptyset$) in $l(h_j)$, therefore by Lemma \ref{lem:blocking}, there is no type 1 (respectively, type 2) blocking pair. This establishes the stability of matching $M'$. Next, consider a student $s_k \in (L - L')$ and suppose she is assigned to school $h_j$. By the definition of $L'$, $h_j$ prefers $\Barr(h_j)$ to $s_k$, therefore, $(\Barr(h_j), h_j)$ form a blocking pair, which is of type 2 if $\Barr(h_j) = \emptyset$ and type 1 otherwise. Hence the matching found in round $\cR_2$ is the largest stable extension of $M$.

% For each student $s_i \in L'$, the Barriers are defined independent of her preference list and she is assigned to the best school $h_j$ in which she appears before $\Barr(h_j)$. Therefore, the best she can do is to reveal her true preference list. Hence, $\cM_2$ is DSIC for students. 
\end{proof}

For the problem of admitting fewer students, stated in Section \ref{sec:S1}, we give the following:

\begin{proposition}
	Let $k$ be the total number of students added from $L$ in round $\cR_2$ in the previous theorem and let $k' < k$. There is a polynomial time mechanism $\cM_2$ that is stability-preserving,  and extends matching $M$ to $M'$ so that $|M'| - |M| = k'$.
\end{proposition}

\begin{proof}
	Let $c'$ denote the capacities of schools after round $\cR_2$ as per Theorem \ref{thm:S1}. Note that the total difference in capacities $c' - c$ over all schools is $k$, where $c$ is the capacity function in round $\cR_1$. Starting with $c'$, arbitrarily decrease the capacities of schools to obtain capacity function $c''$ so that for any school $h_j$, $c'(h_j) - c(h_j) \geq c''(h_j) - c(h_j) \geq 0$ and the total of $c'' - c$ over all schools is $k'$. Starting with $M$, the round $\cR_1$ matching, run the Gale-Shapley algorithm with students from $L$ proposing and with current capacities fixed at $c''$. 
	
	We claim that when this algorithm terminates, the matching found will be student-optimal, stable and each school $h_j$ will be allocated $c''(h_j)$ students from $L$. To see the last claim, observe that the proposals received by any school $h_j$ will be weakly better than the $c'(h_j) - c(h_j)$ students of $L$ who were allocated to $h_j$ under matching $M'$.
\end{proof}

\subsection{Setting A2}
\label{sec:mS2}

The mechanism for round $\cR_2$ for $\textsc{$Min_N Max_L$}$ in Setting A2 is given in Figure \ref{alg:S2,1}. Suppose there is a school $h_j$, student $s_k \in S_M$ is assigned to it and there is a student $s_i \in N$ such that $h_j$ prefers $s_i$ to $s_k$. Now, if $s_i$ is kept unmatched, $(s_i, h_j)$ will form a blocking pair of type 1 by Lemma \ref{lem:blocking}. Next suppose $h_j$ is under-filled and there is a student $s_i \in N$ such that $h_j$ and $s_i$ prefer each other to $\emptyset$. This time, if $s_i$ is kept unmatched, $(s_i, h_j)$ will form a blocking pair of type 2 by Lemma \ref{lem:blocking}. Motivated by this, for a student $s_i$, define the {\em set of schools forming blocking pairs with $s_i$}, denoted $\PgSch(s_i)$, to be:
\[ \PgSch(s_i) = \{h_j \in H ~|~  \mbox{$h_j$ prefers $s_i$ to \LPSt($h_j$)},  \text{ and }   s_i 
				\text{ prefers $h_j$}   \text{ to $\emptyset$} \} \bigcup \]
		\[\hspace*{4cm} \{ h_j \in H ~|~  \mbox{$h_j$ is under-filled and $h_j$ and $s_i$ prefer each other to $\emptyset$}  \}. \]

Therefore, all students in $N'$, computed in Step 3, need to be matched. Our mechanism keeps all students in $N - N'$ unmatched, thereby minimizing the number of students matched from $N$.

We next describe the various barriers that need to be defined. The first one, defined in Step 2, plays the same role as that in Figure \ref{alg:S1}. As before, if $h_j$ is under-filled, $\Barr1(h_j) = \emptyset$. If a student $s_i \in (N' \cup L')$ appears after $\Barr1(h_j)$ in $l(h_j)$ and is assigned to $h_j$, then $(\Barr1(h_j), h_j)$ will form a blocking pair. The second one, $\Barr2(h_j)$ in $(N - N')$ defined in Step 4. Again, if $s_i \in (N' \cup L')$ appears after $\Barr2(h_j)$ in $l(h_j)$ and is assigned to $h_j$, then $(\Barr2(h_j), h_j)$ will form a blocking pair. In step 5, $\Barr(H_j)$ is defined to be the more stringent of these two barriers.

\begin{figure}
	\begin{wbox}
		\textsc{$Min_N Max_L$}$(M,N,L)$: \\
		\textbf{Input:} Stable matching $M$, and sets $N$ and $L$.  \\
		\textbf{Output:} Stable, IC, $Min_N Max_L$ extension of $M$. \\ 
		\begin{enumerate}
			\item $\forall s_i \in S_M: M'(s_i) \leftarrow M(s_i)$ \\			
			\item $\forall h_j \in H: \ \Barr1 (h_j) \lla \ \BStP (h_j)$. \\

			\item $N' \lla \{s_i \in N ~|~ \PgSch(s_i) \ \mbox{is non-empty} \}$. \\

			\item $\forall h_j \in H: \ \Barr2 (h_j) \lla \ \mbox{Best student for} \ h_j \ \mbox{in} \ (N - N')$. \\
			 
		    \item $\forall h_j \in H: \ \Barr (h_j) \lla \ \mbox{Best student for} \ h_j \ \mbox{in} \ \{\Barr1(h_j), \Barr2(h_j) \}$. \\
		    
		    \item $L' \lla \{ s_i \in L  ~|~  \exists h_j \ \mbox{s.t. $s_i$ appears before $\Barr (h_j)$ in $l(h_j)$}$, \\
			\hspace*{2.4cm} $\mbox{and $h_j$ appears before $\emptyset$ in $l(s_i)$} \}$. \\
		    
			\item $\forall s_i \in (N' \cup L'): \ \FSch(s_i) \lla \{ h_j ~|~  s_i \ \mbox{appears before} \ 
                \Barr(h_j) \ \mbox{in} \ l(h_j) \}$. \\
			 \item $\forall s_i \in (N' \cup L'): \  M'(s_i) \lla \ \mbox{Best school for} \ s_i \ \mbox{in} \ \FSch(s_i)$. \\
			 
			 \item $\forall s_i \in ((L - L') \cup (N - N')): \ M'(s_i) \lla \emptyset$. \\

            \item Return $M'$. 
			\end{enumerate}
	\end{wbox}
	\caption{Mechanism for round $\cR_2$ for $\textsc{$Min_N Max_L$}$ in Setting A2}
	\label{alg:S2,1} 
\end{figure} 

The final question is which school should $s_i \in N'$ be matched to? One possibility is to compute for each student $s_i$ the set
\[ T(s_i) = \{h_j \in H ~|~  \exists s_k \ \mbox{s.t. 
			    $M(s_k) = h_j$, $h_j$ prefers $s_i$ to $s_k$, and $s_i$ prefers $h_j$ to $\emptyset$} \}, \]
	and match $s_i$ to her best school in $T(s_i)$.

Assume that $s_i$ is matched to $h_j$ under this scheme. A blocking pair may arise as follows: Assume $s_i$ prefers school $h_k$ to $h_j$ (of course, $h_k \notin T(s_i)$), some student $s_l \in L'$ has been assigned to $h_k$ and $h_k$ prefers $s_i$ to $s_l$. If so, $(s_i, h_k)$ will form a blocking pair. One remedy is to redefine the barrier for $h_k$ so $s_l$ is not assigned to $h_k$. However, this will make the barrier more stringent and the resulting mechanism will, in general, match fewer students from $L$ than our mechanism. The latter is as follows: simply match $s_i$ to the best school which prefers her to the Barrier of that school.

\begin{proof} {\em of Theorem \ref{thm:S2}:} \
The arguments given above already establish stability of matching $M'$ computed. 
% Next, let us argue that the mechanism is DSIC for students in $N$ and $L$. The matching $M$ is not affected by the preference lists of $N$. Therefore the choice of $N'$, and hence $(N - N')$, is independent of the preference lists of $N$. Barrier1 is influenced only by preference lists of $S_M$ and Barrier2 by those of $(N - N')$. Hence Barrier is independent of the preference lists of $N'$ and $L'$. Hence, the matching of students in these two sets is also done in a DSIC manner.

Clearly, each student in $N'$ must be matched because otherwise she forms a blocking pair w.r.t. $M$. Since our mechanism does not match any more students from $N$, it achieves $\textsc{$Min_N$}$. As argued above, not imposing the more stringent of the two barriers computed may result in a blocking pair. Therefore our mechanism imposes the minimum restrictions needed for stability when it is attempting to match students from $L$. Hence it achieves $\textsc{$Min_N Max_L$}$.
\end{proof}

Next, we turn to a slightly different problem within Setting A2, namely find the largest subset of $(N \cup L)$ that can be matched in a stability-preserving manner. We call this problem $\textsc{$Max_{N \cup L}$}$. As shown below, this mechanism also solves the problems $\textsc{$Max_N Max_L$}$ and $\textsc{$Max_L Max_N$}$, namely maximizing the number of students matched from $L$ after having maximized the number of students matched from $N$ and vice versa.

\begin{theorem}
\label{thm:S2,2}
	There is a polynomial time mechanism $\cM_2$ that finds the largest subset of $(N \cup L)$ that can be matched to schools and added to the current matching while maintaining stability. This mechanism also solves $\textsc{$Max_N Max_L$}$ and $\textsc{$Max_L Max_N$}$.
\end{theorem} 

\begin{proof}
We will show that the mechanism presented in Figure \ref{alg:S1}, with $(N \cup L)$ playing the role of $L$, suffices. Barriers for schools are computed as before in Step 2. Denote the subset of $(N \cup L)$ that is matched in round $\cR_2$ by $(N \cup L)'$; it consists of students $s_i \in (N \cup L)$ such that some school $h_j$ prefers $s_i$ to $\Barr(h_j)$ and $s_i$ prefers $h_j$ to $\emptyset$. If so, $s_i$ is assigned to the best such school. 

The argument given in Theorem \ref{thm:S1} suffices to show stability of the matching produced. As before, matching any student from the rest of $(N \cup L)$ will lead to a blocking pair, and hence the mechanism maximizes the number of students matched in round $\cR_2$. 

Finally, since this mechanism acts on $N$ and $L$ independently of each other, it solves $\textsc{$Max_N Max_L$}$ and $\textsc{$Max_L Max_N$}$ as well.
\end{proof}

%%%%%%%%%%%%%%%%%%%%%%%%%%%%%%%%%%%%%%%%%%%%%%%%%%%%

%% file: enumerate.tex
\subsection{Enumeration of Stable Extensions}
\label{sec.enumerate}

In this section we show how to enumerate all the possible stable extensions of a given stable matching with polynomial delay between any two enumerated matchings. 
Specifically, the algorithm takes as input a stable matching $M$ from $S$ to $H$ satisfying capacity $c$ and a set of new students $N = \{s_1, s_2 \ldots s_k\}$ that can be added to the schools. Here the preference lists of all schools and students are also given. The algoirthm enumerates all solutions $M'$ from $S \union N$ to $H \union \{ \emptyset\}$ such that: 
\begin{itemize}
	\item all assignments in $M$ are preserved in $M'$, and
	\item $M'$ is stable with respect to capacity $c'$ where 
	\begin{align} 
	\label{eq:capacity}
	c'(j) = 
	\begin{cases}
	\left| M'^{-1}(h_j) \right|  & \text{ if } \left| M'^{-1}(h_j) \right| > c(j),\\
	c(j)  & \text{ otherwise. } \\
	\end{cases}
	\end{align}
\end{itemize}
Note that $M'^{-1}(h_j)$ is the set of students assigned to $h_j$ under $M'$. We say that $M'$ is a \emph{stable extension of $M$ with respect to $N$}.

\begin{figure}
	\begin{wbox}
		\textsc{StableExtension}$(M,c,N)$: \\
		\textbf{Input:} Stable matching $M$, capacity $c$, new students $N = \{s_1, s_2 \ldots s_k\}$.  \\
		\textbf{Output:} Stable extensions of $M$, with polynomial delay. \\ 
		\newline
		$M_0 \leftarrow M$\\ \\
		$A_1$ = \textsc{FeasibleAssignment}$(M_0, c, s_1)$\\ \\
		For $i_1$ in $A_1$: \\ \\ 
		\mbox{\qquad}$M_1$ $\leftarrow$ Starting from $M_0$, match $s_1$ to $i_1$. \\ \\
		\mbox{\qquad}$A_2$ = \textsc{FeasibleAssignment}$(M_1, c, s_2)$. 			\\ \\
		\mbox{\qquad}For $i_2$ in $A_2$: \\ \\ 
		\mbox{\qquad \qquad} $\vdots$ \\ \\ 
		\mbox{\qquad \qquad \qquad} $A_k$ = \textsc{FeasibleAssignment}$(M_{k-1}, c, s_k)$.  \\ \\
		\mbox{\qquad \qquad \qquad} For $i_k$ in $A_k$: \\ \\ 
		\mbox{\qquad \qquad \qquad \qquad} $M_k$ $\leftarrow$ Starting from $M_{k-1}$, match $s_k$ to $i_k$. \\ \\ 
		\mbox{\qquad \qquad \qquad \qquad} Enumerate $M_k$.

	\end{wbox}
	\caption{Algorithm for enumerating stable extensions of $M$.}
	\label{alg:enumerate} 
\end{figure} 

\begin{figure}
	\begin{wbox}
		\textsc{FeasibleAssignment}$(M_e, c, s_i)$: \\
		\textbf{Input:} Stable matching $M_e$, capacity $c$, student $s_i$.  \\
		\textbf{Output:} Set $A_i$ of all possible assignments for $s_i$. Adding any assignment in $A_i$ to $M_e$ preserves stability. \\ 
		\begin{enumerate}
		\item Initilize $A_i$ to the empty set.  \\ 
		\item For each $h$ in $l(s_i)$, in decreasing order of preferences, do: \\ 
		\begin{enumerate}
		\item If $h = \emptyset$ then Return $A_i \union \{\emptyset\}$. \\ 
		\item Else $h = h_j$: \\ 
		\begin{enumerate}
		\item If $\left| M_e^{-1}(h_j) \right| < c(j)$ then Return $A_i \union \{h_j\}$. \\ 
		\item If $s_i$ appears before $\LPSt(h_j)$ then Return $A_i \union \{h_j\}$. \\ 
		\item If $s_i$ appears after $\LPSt(h_j)$ and before $\BStP(h_j)$ then \\
		 $A_i \leftarrow A_i \union \{h_j\}$.
		\end{enumerate}
		\end{enumerate}
		\end{enumerate}
	\end{wbox}
	\caption{Algorithm for finding feasible matches of $s_i$ w.r.t. current matching $M_c$.}
	\label{alg:findFeasibleMatch} 
\end{figure} 

%At a high level, the algorithm maintains a set $F$ of students such that their assignments (possibly to $\emptyset$) are fixed. 
%Let $M_c$ be the current matching consisting of all fixed assignment. 
%At each step, a student is added to $F$ such that her assignment is compatible with the fixed assignments in $F$.
%In other words, adding the assignment to $M_c$ gives a stable matching from $S \union F$ to $H \union \{\emptyset\}$. 
%When $F = N$, the current matching is returned. The algorithm then backtracks to a previous state and continue. 

%To be precise, at the beginning $F$ is empty and $M_c$ is $M$. 
%Then all possible assignments of $s_1$ is found. $s_1$ is assigned to the first school (or $\empty$). Note that there 

The complete algorithm \textsc{StableExtension}$(M,c,N)$ is given in Figure~\ref{alg:enumerate}.
At a high level, the algorithm maintains a stable extension $M_e$ of $M$ 
with respect to a subset $N'$ of $N$. At each step, a student $s_i$ is added to $N'$ 
and all possible assignments $A$ of $s_i$ that are compatible to $M_e$ are identified. 
In other words, adding each assignment in $A$ to $M_e$ gives a stable extension of $M$ with respect to $N' \union \{s_i\}$. 
The algorithm branches to an assignment in $A$ and continues to the next student.
When $N' = N$, the current matching is returned. The algorithm then backtracks to a previous branching point and continues.

Figure~\ref{alg:findFeasibleMatch} gives the subroutine for finding compatible assignments. 
It takes on input the current matching $M_e$, capacity $c$, student $s_i$ 
and finds all possible assignments $A_i$ of $s_i$ to $H \union \{\emptyset\}$ such that stability is preserved.
Initially, $A_i$ is set to be an empty set. 
The subroutine then goes through the preference list of $s_i$ one by one in decreasing order.
The considered school $h$ is added to $A_i$ and the subroutine terminates if at least one of the following happens: 
\begin{itemize}
	\item $h$ is $\emptyset$,
	\item $h$ is under-filled,
	\item $h$ prefers $s_i$ to $\LPSt(h)$ with respect to $M_e$.
\end{itemize} 
Notice that in the last two scenarios above, if $s_i$ was assigned to any school after $h$ in her preference list, $(s_i,h)$ would form a blocking pair. 
Assume none of the above scenarios happens. The subroutine adds $h$ to $A$ and continues if $h$ prefers $s_i$ to $\BStP(h)$.
%Here assigning $s_i$ to $h$ preserves stability according to Lemma~\ref{lem:blocking}.
Otherwise, $h$ prefers $\BStP(h)$ to  $s_i$. Hence, assigning $s_i$ to $h$ would create a blocking pair. The subroutine continues to the next school in this case. 
The following lemma says that \textsc{FeasibleAssignment} correctly finds all possible assignments of a student, given the current matching, at each step.

\begin{lemma}
	\label{lem:findCorrect}
	Let $N'$ be the set of students assigned (possibly to $\emptyset$) in $M_e$, i.e., $M_e$ is a stable extension of $M$ with respect to $N'$.
	\textsc{FeasibleAssignment}$(M_e, c, s_i)$ finds all possible assignments of $s_i$ to $H \union \{\emptyset\}$ such that 
	%stability is preserved.
	adding each assignment to $M_e$ gives a stable extension of $M$ with respect to $N' \union \{s_i\}$.
\end{lemma}

\begin{proof}
	By the above argument, 	\textsc{FeasibleAssignment} considered all possible assignments. It suffices to prove that adding each to $M_e$ gives a stable extension.
	
	\textsc{FeasibleAssignment} only assigns student $s_i$ to school $h_j$ whenever $s_i$ is before $\BStP(h_j)$ in the preference list of $h_j$. Moreover, whenever $s_i$ is before $\LPSt(h_j)$, no school after $h_j$ in the list of $s_i$ is considered. By Lemma~\ref{lem:blocking}, no blocking pair of type 1 is created. 
	
	Similarly, when $h_j$ is under-filled, no school after $h_j$ in the list of $s_i$ is considered. Hence, no blocking pair of type 2 is created according to Lemma~\ref{lem:blocking}.
	
	Finally, whenever $h_j$ is filled, the capacity $c(j)$ increases by 1 according to (\ref{eq:capacity}). Hence, adding the assignment of $s_i$ to $h_j$ gives a stable matching with respect to the new capacity. 
\end{proof}

Another important observation is that there is at least one possible assignment 
returned by \textsc{FeasibleAssignment}$(M_e, c, s_i)$ for any input of the subroutine. To see this, 
consider two cases: 
\begin{enumerate}
	\item if all schools $h_j$ that appear before $\emptyset$ in the preference list of $s_i$ prefer $\BStP(h_j)$ to $s_i$, then $\emptyset$ is a possible assignment for $s_i$,
	\item if at least one school $h_j$ that appears before $\emptyset$ in the preference list of $s_i$ prefers $s_i$ to $\BStP(h_j)$, 
	$h_j$ is a possible assignment.
\end{enumerate}
Hence, we have the following lemma:

\begin{lemma}
	\label{lem:findReturn}
	\textsc{FeasibleAssignment}$(M_e, c, s_i)$ returns at least one possible assignment.
\end{lemma}

From Lemmas~\ref{lem:findCorrect} and \ref{lem:findReturn}, we can prove the main theorem of this section:

\begin{theorem}
	\textsc{StableExtension}$(M,c,N)$ enumerates all possible stable extension of $M$ with respect to $N$. Moreover, the time between any two enumerations is $O((k+n)m)$. 
\end{theorem}
\begin{proof}
Let $M'$ be a stable extension of $M$ with respect to $N$. Let $i_1, i_2 \ldots i_k$ be the assignment of $s_1, s_2 \ldots s_k$ respectively. 
For $1 \leq l \leq k$, denote by $M_l$ the matching obtained by adding $s_1i_1, s_2i_2, \ldots s_li_l$ to $M$.
By Lemma~\ref{lem:findCorrect}, $i_{l+1}$ must be in the set of possible assignments returned by \textsc{FeasibleAssignment}$(M_l, c, s_{l})$. 
Hence, \textsc{StableExtension}$(M,c,N)$ correctly enumerates $M_k = M'$ at some point.

By Lemma~\ref{lem:findReturn}, there are at most $k$ calls of \textsc{FeasibleAssignment} between two matchings enumerated by \textsc{StableExtension}. Each call of \textsc{FeasibleAssignment} goes through a student preference list of at most $m+1$ schools (including $\emptyset$). The time needed for initializing and updating  $\LPSt(h_j)$ and $\BStP(h_j)$ for each school $h_j$ at each step is $O((k+n)m)$. Hence the total time is $O((k+n)m)$.
\end{proof}

%% file: B1.tex
\section{Mechanisms for Type B Settings}
\subsection{Setting B1}
 We first show some structural properties of minimum stable re-allocations in this setting. The mechanism $\cM_2$ for adding new students is presented in Figure $\ref{alg:S3}$. $N' \subseteq N$, defines the set of students who form blocking pairs with the current matching $M$. $SM$ denote the set of stable matchings over the instance $ I = (S \union N, H, c)$, and $MSR$ represents the set of all minimum stable re-allocations of $M$. For all $s_i \in N$ we set $M(s_i) = \emptyset$

\begin{lemma}\label{lem:student_worse}
Each student $s_i \in S$ does weakly worse in any minimum stable re-allocation.
\end{lemma}

\begin{proof}
This follows directly from the fact that $M$ is a student-optimal stable matching.
\end{proof}

\begin{definition}
$s_i \in S \union N$ is \textbf{moved} in $M' \in MSR$, if $M(s_i) \neq M'(s_i)$
\end{definition}

\begin{lemma}\label{lem:same-move}
All minimum stable re-allocations of $M$ move the same set of students, $S_R$.
\end{lemma}

\begin{proof}
For all $s_i \in L \union (N - N')$, $s_i$ must come after $\LPSt(h_j)$ for all schools $h_j$. Since the capacities of schools are fixed, these students can never be matched in any stable matching.
Let $M'$ be some minimum stable re-allocation of $M$ with a student, $s_i \in S_M$, such that $M'(s_i) = \emptyset$. Since $M'$ is a stable matching, by the Rural Hospitals Theorem $s_i$ will remain unmatched in all stable matchings over updated instance. Similarly the students $s_i \in N'$ who do get matched in $M'$, will be matched in all stable matchings.

Finally, let there be an $s_i \in S_M$ and $M', M''$ be two minimum stable re-allocations such that $M(s_i) = M'(s_i) \neq M'(s_i)$. Then the following include all possibilities for $s_i$
    \begin{itemize}
		\item $S_{1} = \{s_i \in S_M \ | \ M(s_i) = M'(s_i) \neq M''(s_i) \}$
		\item $S_{2}	= \{ s_i \in S_M \ | \ M(s_i) = M''(s_i) \neq M'(s_i) \}$
		\item $S_{3} = \{ s_i \in S_M \ | \ M(s_i) \neq M'(s_i) \textnormal{ and } M(s_i) \neq M''(s_i) \}$
		\item $S_{4} = \{s_i \in S_M \ | \ M(s_i) =  M'(s_i) = M''(s_i) \}$
	\end{itemize}
	
Now if we consider the matching $M_U = M' \vee M''$(the school they prefer more)  this will send all $s_i \in S_{1}, S_{2}, S_{4}$ to their original partners in $M$, since from Lemma $\ref{lem:student_worse}$ we have that each $s_i$ must do weakly worse in any $MSR$. From Theorem $\ref{sec:stable-matching-lattice}$, $M_U$ is the join of two stable matchings and hence a stable matching.\\
The number of re-allocations in $M', M'', M_U$ is $|S_{3}| + |S_{2}|, \ |S_{3}| + |S_{1}|$ and 
$M_U = |S_{3}|$, respectively.
By assumption $S_{1}$ is non-empty and thus $M_U$ requires fewer re-allocations than $M''$, a contradiction. Thus if $s_i$ is moved in one re-allocation it is moved in all re-allocations.

\end{proof}

\begin{corollary}\label{cor:S_F-fixed}
All students $s_i \in S \union N - S_R$ are matched to the same school in all minimum stable  re-allocations.
\end{corollary}

We denote these students as $S_F$, and let $M_F$ represent the matching restricted to these students. Then for all $M' \in MSR$ and $s_i \in S_F$, $M_F(s_i) = M(s_i) = M'(s_i)$.
Let $H_R$, be the set of schools that students in $S_R$ are matched to in some minimum stable re-allocation then we have:

\begin{corollary} \label{cor:same-H}
All minimum stable re-allocations will match students in $S_R$ to schools in $H_R$. Moreover, if $k$ students from $S_R$ are matched to a school $h_j \in H_R$, then all minimum stable re-allocations will have $k$ students from $S_R$ matched to $H_R$.
\end{corollary}

% \begin{corollary} \label{cor:k-same-H}
% If $k$ students from $S_R$ are matched to a school $h_j \in H_R$, then all minimum stable re-allocations will have $k$ students from $S_R$ matched to $H_R$.
% \end{corollary}

Consider the stable matching instance $I'$, defined below:
\begin{itemize}

    \item $\forall s_i \in S_R, \ \Barr(s_i) = $ Most preferred $h_j \in H - H_R$, such that $h_j$ is currently under-filled and prefers $s_i$ to $\emptyset$, or $h_j$ prefers $s_i$ to $\LPSt(h_j)$.
    
    \item $\forall h_j \in H_R, \ \Barr(h_j) = \BStP(h_j)$ among students in $S_F$
    
    \item $\forall s_i \in S_R$ $l'(s_i) = l(s_i)$. Place the $\emptyset$ to the immediate left of $\Barr(s_i)$
    
    \item $\forall h_j \in H_R$ $l'(h_j) = l(h_j)$. Place the $\emptyset$ to the immediate left of $\Barr(h_j)$
    
    \item Let $M'$ be some $MSR$, then $c'(h_j) = |\{s_i \in S_R \ | \ M'(s_i) = h_j  \}|$
    
    \item $I' = (S_R, H, c')$ with preference lists $l'(s_i), l'(h_j)$ defines a stable matching instance.

% 	\item $\forall s_i \in H_R, \ \Barr(h_j) = $ Most preferred $s_i \in S_F$ in $h_j$ preference list such that $(s_i, h_j) \subseteq E'$, and $s_i$ prefers $h_j$ to $M(s_i)$.
	
% 	\item $\forall s_i \in S_R, \ \Barr(s_i) = $ Most preferred $h_j \in B_F$ in $h_j$ preference list such that $(s_i, h_j) \subseteq E'$, and $h_j$ prefers $s_i$ to $M(h_j)$.

% 	\item $E'' = \{(s_i, h_j) \in E' \ | \ s_i \in S_R, h_j \in H_R,$ $h_j$ prefers $s_i$ to $\Barr(h_j)$, and $s_i$ prefers $h_j$ to $\Barr(s_i) \}$ 

	%\item $G' = (S_R, H_R, E'')$
	
\end{itemize}

\begin{lemma}\label{sec:allowed-edges}
If a student (resp. school) in $S_R$ (resp. $H_R$) comes after the barrier of $h_j \in H_R$ (resp. $s_i \in S_R$) then this student (resp. school) cannot be matched to $h_i$ (resp. $s_i$) in any minimum stable re-allocation.
\end{lemma}

\begin{proof}
	Let $M'$ be a minimum stable re-allocation where such a pair $(s_i, h_j)$ exists. By assumption $s_i$ must be to the right of $\Barr(h_j)$ on $h_j$'s preference list, or $h_j$ must be to the right of $\Barr(s_i)$ on $s_i$'s preference list.\\
    $\Barr(h_j) \in S_F$ and so its partner is fixed in all $MSR$ by Corollary $\ref{cor:S_F-fixed}$. $\Barr(h_j)$ prefers $h_j$ to its partner in $M$, and if $s_i$ is  to the right of $\Barr(h_j)$ on $l'(h_j)$ then $(\Barr(h_j), h_j)$ forms a blocking pair in $M'$. Similarly if $h_j$ is to the right of $\Barr(s_i)$ on $'l(s_i)$ then $(s_i, \Barr(s_i))$ would form a blocking pair in $M'$, either way a contradiction.

\end{proof}

Let $SM_{I'}$ denote the set of all stable matchings over $I'$.

\begin{lemma}\label{sec:compose-lemma}
	$\forall \ M_{I'} \in SM_{I'},\ M'= M_{I'} \union M_F$ is a minimum stable re-allocation.
\end{lemma}

\begin{proof}
	We only need to show that $M'$ is stable. Let $(s_i, h_j)$ be a blocking pair of type 1, then $s_i \in S_F$ and $h_j \in H_R$, or $s_i \in S_R$ and $h_j \in H - H_R$. Let us first consider the former and let $\LPSt(h_j) = s_j$ in $M'$. $h_j$ must prefer $s_i$ to $s_j$. Then $s_i$ is ranked atleast as high as $\Barr(h_j)$ and $s_j$ comes after the Barrier on $h_j's$ preference list. By construction this can never happen, and we get a contradiction. A similar argument works for the case when $s_i \in S_R$ and $h_j \in H- H_R$.
	If $(s_i, h_j)$ is a blocking pair of type 2 then $(s_i, h_j)$ would form a blocking pair in $M$.
	
\end{proof}

\begin{lemma} \label{sec:decompose-lemma}
	Any $M' \in MSR$ can be decomposed into $M_{I'} \union M_F$, where $M_{I'} \in SM_{I'}$.

\end{lemma}

\begin{proof}
	Let $M' \in MSR$, by Corollary $\ref{cor:S_F-fixed}$, $M_F \subseteq M'$. By Corollary $\ref{cor:same-H}$, students in $S_R$ must be matched to students in $H_R$. Restricting $M'$ to students in $S_R$, there must be a stable matching over $(S_R, H_R, c')$ with preference lists $l(s), l(h)$. By Lemma $\ref{sec:allowed-edges}$, the students and schools moved to the right of $\emptyset$ in $l'(s_i)$ and $l'(h_j)$ can never be matched to $s_i$ or $h_j$ in any $MSR$. Hence $M'$ restricted to students in $S_R$ must be a stable matching over $I'$ 

\end{proof}

\begin{lemma} \label{lem:msr-lattice}
	$(MSR, \succeq)$ defines a sublattice of $(SM, \succeq)$.
\end{lemma}

\begin{proof}
Let $M_1, M_2 \in MSR$. Define $M_U = M_1 \vee M_2$, and $M_L = M_1 \wedge M_2$, where for each $s_i, \ M_U(s_i)$ is the school it weakly prefers in $M_1(s_i), M_2(s_i)$ and $M_L(s_i)$ is the complement of $M_U(s_i)$.\\
By Lemma $\ref{sec:decompose-lemma}$, we can decompose $M_1$ as $M_{I'}^{1} \union M_F$, and $M_2$ as $M_{I'}^{2} \union M_F$ where $M_{I'}^{1}$ and $M_{I'}^{2} \in SM_{I'}$. For $s_i \in S_F$, $M_1 \vee M_2 = M_1 \wedge M_2 = M_F$. 
For $s_i \in S_R$ since $M_{I'}^{1}, M_{I'}^{2} \in S_{I'}$ by Theorem $\ref{sec:stable-matching-lattice}$, $ \ M_{I'}^{1} \vee M_{I'}^{2}$ and $M_{I'}^{1} \wedge M_{I'}^{2}$ are also stable matchings in $SM_{I'}$. Then from Lemma $\ref{sec:compose-lemma}$ $M_U = M_{I'}^{1} \vee M_{I'}^{2} \union M_F$ is a MSR, and $M_L = M_{I'}^{1} \wedge M_{I'}^{2} \union M_F$ is also MSR. Thus $M_L$ and $M_U$ are the meet and join of two minimum stable re-allocations. 
\end{proof}

The lattice structure defined above leads to the following definitions:

\begin{definition} \textbf{Student-Optimal (School-Pessimal) Minimum Stable Re-allocation:}\\
Each member $s_i \in S_R$ is matched to its best possible partner among all minimum stable re-allocations.\\
\end {definition}

\begin{definition} \textbf{Student-Pessimal (School-Optimal) Minimum Stable Re-allocation:}\\
Each member $s_i \in S_R$ is matched to its worst possible partner among all minimum stable re-allocations\

\end{definition}

\begin{figure}[htbp]
	\begin{wbox}
		\textsc{Adding New Students}$(M,N)$: \\
		\textbf{Input:} Stable matching $M$ and set $N$.   \\
		\textbf{Output:} Minimum stable re-allocation of $M$. \\ 
		\begin{enumerate}

			\item $\forall s_i \in S_M: M'(s_i) \leftarrow M(s_i)$ \\	
			
			\item While $\exists s_i$  unmatched and $(s_i, h_j) $ form a blocking pair do\\
			    \begin{enumerate}
			        \item $h \leftarrow $ Best possible $h_i$ in Schools-FBPairs($s_i$)\\
			        
			        \item if h is filled to capacity then unmatch $\LPSt(h)$ \\

			        \item $M'(s_i) \leftarrow h$\\
			    \end{enumerate}
		    
		    \item Return $M'$.\\

			\end{enumerate}

	\end{wbox}
	\caption{Mechanism $\cM_2$ for adding new students in round $\cR_2$}
	\label{alg:S3} 
\end{figure} 

\begin{observation}\label{obs:GS-equal}
$\cM_2$ can be thought of as a continuation of Gale-Shapley with the inclusion of new students, where initially in round $\cR_1$ only students in $S$ propose to schools, and in round $\cR_2$ the students in $N$ also start proposing to schools. Like Gale-Shapley, $\cM_2$ terminates with a maximal matching having no blocking pairs. Since the Gale-Shapley algorithm produces the same matching independent of the order of proposals among students, $\cM_2$ is equivalent to running Gale-Shapley over the whole instance.

\end{observation} 

% \begin{lemma}\label{sec:new_participant_lemma_2}
% If $s_i \in S$ is matched to $h_j$ in any step of $\cM_2$, then $s_i$ is matched to $h_j$ or worse in any minimum stable re-allocation.
% \end{lemma}

% \begin{proof}
% The proof is by induction. At the beginning of $\cM_2$ this is true because of Lemma \ref{lem:student_worse}. Assume during some step $n$, $\cM_2$ matches $s_i$ to $h_j$.
% Let $M^*$ be a minimum stable re-allocation where $s_i$ is matched to someone better than $h_j$, say $h_l$. Then at step $n$, $h_l$ preferred $\LPSt(h_l)$ to $s_i$. By the inductive hypothesis $\LPSt(h_l)$ must be matched to a school worse than $h_l$, which results in a blocking pair.

% Then $h_j$ must prefer its partner in $M^*$ to $s_i$, but this means that $M^*(h_j)$ must have rejected $h_j's$ proposal in $\cM_2$ and at step $n$, $M^*(h_j)$ was matched to someone better, contradicting the inductive hypothesis.\\

% The proof is by induction. At the beginning of $\cM_2$ this follows from Lemma \ref{lem:student_worse}. Assume during some step $n$, $\cM_2$ matches $s_i$ to $h_j$. For any school $h_k \in \PSch(s_i)$, either $h_k$ prefers $\emptyset$ over $s_i$ or $h_k$ is filled to capacity and prefers $\LPSt(h_k)$ over $s_i$. If $s_i$ were to be matched to $h_k$ then a student previously matched to $h_k$ would be matched to a school worse than $h_k$ by the inductive hypothesis, which results in a blocking pair.

% \end{proof}

\begin{proof} {\em of Theorem \ref{thm:S3}:}
We first show that if $s_i \in S$ is matched to $h_j$ in any step of $\cM_2$, then $s_i$ is matched to $h_j$ or worse in any minimum stable re-allocation. The proof is by induction. At the beginning of $\cM_2$ this is true because of Lemma \ref{lem:student_worse}. Assume during some step $n$, $\cM_2$ matches $s_i$ to $h_j$.
Let $M^*$ be a minimum stable re-allocation where $s_i$ is matched to someone better than $h_j$, say $h_l$. Then at step $n$, $h_l$ preferred $\LPSt(h_l)$ to $s_i$. For $M^*$ to be stable then there must be a student matched to $h_l$ at step $n$, that is now matched to a school it prefers more. But this contradicts the inductive hypothesis.
$\cM_2$ thus makes the minimum number of re-allocations.

% \begin{theorem} \label{sec:new_participants_theorem}
% 	$\cM_2$ is a polynomial time mechanism that computes an A-Pessimal minimum stable re-allocation, $M'$, over $M, (A, B \union B', E')$
% \end{theorem}

% \begin{proof}
% 	Suppose $s_i \in A$ such that $M(s_i) \neq M'(s_i)$. By Lemma \ref{sec:new_participant_lemma_2} if $s_i$ is not matched to $M'(s_i)$, or better then $(s_i, M'(s_i))$ forms a blocking pair. Since $\cM_2$ only moves $s_i$ from their original partners when they have to be moved, it performs the minimum number of re-allocations.\\
% 	To see that it is A-Pessimal, let $M^*$ be another minimum stable re-allocation where there is an $s_i$ who prefers $M'(s_i)$ to $M^*(s_i)$. By Lemma $\ref{sec:new_participant_lemma_2}$, $s_i$ must be matched to $M'(s_i)$ or better in any minimum stable re-allocation, a contradiction.\\
	
% 	The proof is by induction. At the beginning of $\cM_2$ this follows from Lemma \ref{sec:new_student_lemma}. Assume during some step of $\cM_2$ that $s_i$ is matched to $h_j$. For any school $h_k \in \PSch(s_i)$, either $h_k$ prefers $\emptyset$ over $s_i$ or $h_k$ is filled to capacity and prefers $\LPSt(h_k)$ over $s_i$. If $s_i$ were to be matched to $h_k$ then a student previously matched to $h_k$ would be matched to a school worse than $h_k$ by the inductive hypothesis, which results in a blocking pair.
	
\end{proof}

\begin{corollary}
$\cM_2$ produces a student-optimal minimum stable re-allocation.
\end{corollary}

\begin{lemma} \label{lem:M_3}
	There exists a mechanism $\cM_3$, that finds a school-optimal minimum stable re-allocation in polynomial time.

\end{lemma}

\begin{proof}
	From Lemma $\ref{sec:decompose-lemma}$, the school-optimal minimum stable re-allocation corresponds to the school-optimal stable matching in $SM_{I'}$. To find a school-optimal stable matching of $SM_{I'}$, run $M_2$ on the whole instance and find the sets $S_R, H_R, M_F$. We can then construct $I'$, and run a school proposing version of Gale-Shapley algorithm on this instance to find the school-optimal matching, $M_{opt}$. Applying Lemma $\ref{sec:compose-lemma}$, taking the union of $M_F$ and $M_{opt}$ will give a school-optimal minimum stable re-allocation.
\end{proof}

%% file: B2.tex
\subsection{Setting B2}

A first approach to finding a minimum stable minimum re-allocation in Setting $B2$ would be to run Gale-Shapley over the whole instance. However unlike Setting B1, Example $\ref{ex:GS_bad_MSR}$ shows that this could require as many as $|S|$ possible re-allocations.

\begin{example}\label{ex:GS_bad_MSR}
Let there be $n+1$ students and schools. The preference lists (mod $n+1$) for any student $s_i$ is $(h_{i-1}, h_{i}, ..., h_{i-2})$ and the preference list for any school $h_j$ is $(s_{j}, s_{j+1}, ..., s_{j-1})$. In round $\cR_1$, all participants but $h_{n+1}$ are present and each school has $1$ seat. In round $\cR_2$, $h_{n+1}$ arrives with capacity $1$. The only stable matching, $M$, from round $\cR_1$ would match each $s_i \rightarrow h_i$ and $s_{n+1}$ would remain unmatched. Assigning $s_{n+1}$ to $h_{n+1}$ would result in a stable matching requiring no re-allocations. However, running Gale-Shapley over all participants would yield the matching where each $s_{i} \rightarrow h_{i-1}$, but this matching requires $n$ re-allocations.\\
\end{example}

% \begin{example}
% Assume there are $n+1$ members in both sets $A$ and $B$. Let the preference lists be as follows:\\

% \begin{table}[hbt!]
% 	\begin{tabular}{lllllllllllllll}
% 		\bm{$a_{1}:$} & $b_{n+1}$  & $b_{1}$ & $b_{2}$ & ....  &  $b_{n}$ & &  &  	\bm{$b_{1}:$} & $a_{1}$  & $a_{2}$ & $a_{3}$ & ....  &  $a_{n+1}$ &  \\

% 		\bm{$a_2:$} & $b_{1}$  & $b_{2}$ & $b_{3}$ & ....  &  $b_{n+1}$ & &  &       \bm{$b_2:$} & $a_{2}$  & $a_{3}$ & $a_{4}$ & ....  &  $a_{1}$ &  \\
% 		\bm{$a_3:$} & $b_{2}$  & $b_{3}$ & $b_{4}$ & ....  &  $b_{1}$ & & &         \bm{$b_3:$} & $a_{3}$  & $a_{4}$ & $a_{5}$ & ....  &  $b_{2}$ &  \\
% 		\vdots & \vdots & \vdots & \vdots & & \vdots  & & &	\vdots & \vdots & \vdots & \vdots & & \vdots   \\

% 		\bm{$a_{n}:$} & $b_{n-1}$  & $b_{n}$ & $b_{n+1}$ & ....  &  $b_{n-2}$ & & & 	\bm{$b_{n}:$} & $a_{n}$  & $a_{n+1}$ & $a_{1}$ & ....  &  $a_{n-1}$ &  \\
% 		\bm{$a_{n+1}:$} & $b_{n}$  & $b_{n+1}$ & $b_{1}$ & ....  &  $b_{n-1}$ & & & 	\bm{$b_{n+1}:$} & $a_{n+1}$  & $a_{1}$ & $a_{2}$ & ....  &  $a_{n}$ &  \\
	
% 	\end{tabular}
% \end{table}

% In Round $\cR_1$ all participants but $b_{n+1}$ are present. In Round $\cR_2$  $b_{n+1}$ arrives. The only stable matching, M, from $\cR_1$ would match each $s_i \rightarrow b_i$ and $a_{n+1}$ would remain unmatched. Assigning $a_{n+1}$ to $b_{n+1}$ would result in a stable matching requiring no re-allocations. However, running Gale-Shapley over all participants would yield the matching where each $a_{i} \rightarrow b_{i-1}$, but this matching requires $n$ re-allocations.\\

% \end{example}

\begin{lemma}\label{lem:weak-improve}
Each student weakly improves in any minimum stable reallocation.
\end{lemma}

\begin{proof}
Suppose $M'$ is a minimum stable re-allocation of $M$, such that some students prefer their match in $M$ to $M'$. Let $W$ denote these set of students. Consider $M^* = M \vee M'$, the matching where each $s_i$ chooses the better of $M(s_i), M'(s_i)$. We show that $M^*$ is stable and satisfies the capacity constraints of schools.\\
To see that $M^*$ satisfies the capacities of schools note that $M^*$  only moves the students in $W$ from $M'$ back to their original schools. If a set of students in $W$ leave a school, $h_j$, in going from $M$ to $M'$, then the set of students replacing them at that school must also be in $W$. This is because $h_j$ must prefer the replacing students to the leaving students. If there is a replacing student, $s_i$, not in $W$ $(s_i, h_j)$ would form a blocking pair in $M$\\
To argue stability, let $(s_i, h_i), (s_j, h_j) \in M^*$, and $(s_i, h_j)$ be a blocking pair of type 1. Then $s_i$ is matched to $h_i$ or worse in both $M$ and $M'$, but $(s_j, h_j)$ must be in one of these matchings contradicting stability. If $(s_i, h_j)$ was a blocking pair of type $1$, then it must be a blocking pair in $M$ or $M'$.\\
$M^*$ is a stable re-allocation of $M$ requiring $|W|$ fewer re-allocations than $M'$, a contradiction.

\end{proof}

The lattice structure given in the previous section carries over to Setting B2 as well. The reason for this is both Settings B1 and B2 can be reduced to an instance where schools and students have unit capacity. In round $\cR_1$ a stable matching is found and in round $\cR_2$ a set of new participants arrive on one side. Since each school has unit capacity, schools and students become interchangeable.

\begin{lemma}
All minimum stable re-allocations of $M$ move the same set of students, $S_R$.
\end{lemma}

\begin{proof}
The proof is similar to Lemma $\ref{lem:same-move}$; to arrive at a contradiction consider the matching $M' \vee M''$
\end{proof}

By including schools in $H'$ in the construction of barriers, we can apply the same analysis from lemma $\ref{lem:msr-lattice}$, to get:

\begin{lemma}
$(MSR, \succeq)$ is a sublattice of $(SM, \succeq)$.
\end{lemma}

\begin{figure}[htbp]
	\begin{wbox}
		\textsc{Adding New Schools}$(M,H')$: \\
		\textbf{Input:} Stable matching $M$ and set $H'$.   \\
		\textbf{Output:} Minimum stable re-allocation of $M$. \\ 
		\begin{enumerate}
			\item $\forall s_i \in S_M: M'(s_i) \leftarrow M(s_i)$ \\	
			
			\item While $\exists h_j \in H \cup H'$ with unmet-capacity and $\BStP (h_j) \neq \emptyset$:\\
			    \begin{enumerate}
			        \item Break current match if exists of $\BStP (h_j)$\\
			        \item $M' \leftarrow M' \cup (\BStP(h_j), h_j)$\\
			    \end{enumerate}
		    
		    \item Return $M'$.\\

			\end{enumerate}

	\end{wbox}
	\caption{Mechanism $\cM_2$ for adding new schools in round $\cR_2$}
	%\label{alg:S1} 
\end{figure}

\begin{proof} {\em of Theorem \ref{thm:S3}:}
We show that if $s_i$ is matched to $h_j$ in any step of $\cM_2$, then $s_i$ is matched to $h_j$ or better in any minimum stable re-allocation.
The proof is by induction. At the beginning of $\cM_2$ this is true because of Lemma \ref{lem:weak-improve}. Assume during some step $n$, $\cM_2$ matches $s_i$ to $h_j$.
Let $M^*$ be a minimum stable re-allocation where $s_i$ is matched to someone worse than $h_j$. Then $h_j$ must prefer its $\LPSt(h_j)$ in $M^*$ to $s_i$. At step $n$, $h_j$ had unmet-capacity and $s_i$ was $\BStP(h_j)$, so there must be a student $s_j$ matched to $h_j$ in $M^*$ who was matched to a school it preferred more at step $n$. This contradicts the inductive hypothesis.
% Suppose $s_i \in S_M$ such that $M(s_i) \neq M'(s_i)$. By Lemma \ref{sec:new_participant_lemma_2} if $s_i$ remains at $M(s_i)$ there will be a blocking pair.
Since $\cM_2$ moves students from their original schools when they have to be moved, it performs the minimum number of re-allocations.
\end{proof}

\begin{corollary}
$\cM_2$ produces a school-optimal minimum stable re-allocation.
\end{corollary}

\begin{lemma}
    There exists a mechanism $\cM_3$, that finds a student-optimal minimum stable re-allocation in polynomial time.
\end{lemma}

\begin{proof} 
    Apply the same procedure outlined in lemma $\ref{lem:M_3}$ on the student-optimal matching over $SM_{I'}$.
\end{proof}

% \begin{lemma}
% 	There exists a mechanism $\cM_3$, that finds an A-optimal minimum stable re-allocation.

% \end{lemma}

% \begin{proof}
% 	From Lemma $\ref{sec:decompose-lemma}$, the A-optimal minimum stable re-allocation corresponds to the A-optimal stable matching in $S_{G'}$. To find an A-optimal stable matching of $S_{G'}$, run $M_2$ on the whole instance and find the sets $A_R, B_R, M_F$. We can then construct $G'$, and run Gale-Shapley on this instance to find the A-optimal matching, $M_{opt}$, in $S_{G'}.$ Applying Lemma $\ref{sec:compose-lemma}$, taking the union of $M_F$ and $M_{opt}$ will give an A-optimal minimum stable re-allocation.\\ 

% \end{proof}

%% file: incentive_compatibility.tex
\section{Incentive Compatibility} \label{sec:incentives}

For the four settings discussed, it would be highly desirable if we could prove that mechanism $\cM_2$ in round $\cR_2$ is DSIC. We show that for Setting B1 that this truly is the case. Unfortunately for Settings $A1, A2$ and $B2$ we show that the current mechanisms outlined above are not incentive compatible. We relax DSIC and consider the weaker notion of of a  mechanism for which {\em incentive compatibility is a Nash equilibrium (ICNE)}. Under such a mechanism, a student cannot gain by misreporting her choices, if all other students are truthful. We show that no mechanisms in round $\cR_2$ for Setting A1, A2 and B2 can be even ICNE.

\begin{lemma}
$\cM_2$ in Setting B1 is DSIC for students.
\end{lemma}

\begin{proof}
Since $\cM_1$ produces a student-optimal matching, from observation $\ref{obs:GS-equal}$, the matching returned by $\cM_2$ is the same as running the Gale-Shapley Algorithm on all participants. DSIC follows from DSIC of the latter.
\end{proof}

%For Settings $A1, A2$ we have the following example:
% \begin{example}
% Let there be 3 students $A,B,C$ and two schools $1,2$. The preference lists of the students are $l(A) = (1,2), \ l(B) = (2,1), \ l(C) =  (2,1)$. The preference lists for the schools are $l(1) =(B, C, A) , \ l(2) =  (A, B, C)$. In round $\cR_1$ each school has a capacity of 1 seat. 
% % Under truthful reporting the matching in round $\cR_1$ would be $(A,2), (B,1)$, and the extension produced by $\cM_2$ in Round 2 would be $(A,2), (B,1), (C,1)$. However if C misreports its preference list as $(2, \emptyset)$, then the matching in round $\cR_1$ would be $(A,1), (B,2)$ and the extension produced by $\cM_2$ in round  $\cR_2$ would be $(A,1), (B,2), (C,2)$. $C$ does better by misreporting its preference list.
% \end{example}
\begin{example}\label{ex:A1}
Let there be two students $A,B$ and two schools $1,2$. The preference lists of the students are $l(A) = (1,2), \ l(B) = (1,2)$. The preference lists for the schools are $l(1) = (A, B), \ l(2) = (B, A)$. In round $\cR_1$ each school has  a capacity of $1$ seat.
\end{example}

\begin{lemma}
There is no pair of stability-preserving, ICNE mechanism $(\cM_1, \cM_2)$ for Setting A1 and A2.
\end{lemma}

\begin{proof}
Consider example $\ref{ex:A1}$. Under truthful reporting there is only one stable matching in round $\cR_1$: $(A,1), (B,2)$. The extension produced by $\cM_2$ in round $\cR_2$ would remain the same since $L$ is empty. However, if B misreports its preference list as $(1, \emptyset)$ then the only stable matching, $M$, in round $\cR_1$ would be $(A,1)$. In round $\cR_2$ any mechanism that finds a maximal extension of $M$ will try to match $B$, the only valid stable extension is $(A,1), (B,1)$.
B benefits from misreporting its preference list.
\end{proof}

%Similarly for Setting B2 we provide the following example:

\begin{example} \label{ex:B2}
Assume there are 2 students $A, B$ and three schools $1,2,3$. The preference lists of the students are $l(A) = (2,1,3), \ l(B) = (1,2,3)$. The preference lists of the schools are $l(1) = (A,B), \ l(2) = (B,A), \ l(3) = (A,B)$. In round $\cR_1$, school $1$ has $1$ seat and in round $\cR_2$ schools $2$ and $3$ arrive with $1$ seat each. 
% Under truthful reporting $(A,1)$ will be assigned in round $\cR_1$, and thus assigning $(B,2)$ in round $\cR_2$ will be the only stable matching with no re-allocations. However if $A$ misreports her preference as $(2,3,1)$ then the only stable matching in round $\cR_2$ is $(A,2),(B,1)$ which is preferred by $A$.
\end{example}

\begin{lemma}
There is no pair of stability-preserving, ICNE mechanism $(\cM_1, \cM_2)$ for Setting B2.
\end{lemma}

\begin{proof}
Consider example $\ref{ex:B2}$. $(A,1)$ is the only valid stable matching after round $\cR_1$. Under truthful reporting assigning $(B,2)$ in round $\cR_2$ will be the only stable matching with no re-allocations. However if $A$ misreports its preference as $(2,3,1)$ then the only stable matching in round $\cR_2$ is $(A,2),(B,1)$ which is preferred by $A$.

\end{proof}

The key distinction for incentive compatibility between Setting A1,A2 and B1, is that in B1 if a student is unmatched after round $\cR_1$ it will remain unmatched after round $\cR_2$. However in Settings A1, A2 we try to accommodate students who were unmatched after round $\cR_1$, so they still have a chance to get matched in $\cR_2$. This allows them the possibility of affecting the matching produced in $\cR_1$ by misreporting their preference list so as to make the Barriers computed by $\cM_2$ more favorable for them.

%% file: hardness.tex
\section{NP-Hardness Results}
\label{sec.hardness}

We show that several natural problems for the Settings discussed above are NP-Hard. The first three problems involve different variations of Setting A2. Problem $4$ asks if there is a way to re-allocate students who were matched in round $\cR_1$ so as to increase the number of students matched in round $\cR_2$. For Problem $5$ we deviate from the two round setting and go back to the single round setting over schools and students with an added weight function defined over edges between students and schools. Problem $5$ asks to define a capacity vector over schools and a corresponding stable matching that maximizes the total weight. We define all problems formally below.

\begin{problem}
     A set of new students $N$ arrive in round $\cR_2$. Let $L$ be the set of students in round $\cR_1$ who are unmatched in $M$. The City wants to extend original matching M in a stability-preserving manner so that it maximizes the number of students who get matched from $L$, and subject to this, minimize the number of students who get matched from $N$. ($max_L min_N$)
   
\end{problem}    

\begin{problem}
    Same setting as Problem 1, but the City wants to maximize the number of students who get matched from $N$, and subject to this, minimize the number of students who get matched from $L$. ($max_N min_L$)

\end{problem}

\begin{problem}
    A set of new students $N$ arrive in round $\cR_2$. The City wants to extend the matching in a stability-preserving manner to include $k$ students from $N$, such that it maximizes the number of students matched from $L$.  It can be assumed without loss of generality that $k$ is large enough to allow for a stable extension. $(k$-$max_L)$
    
\end{problem}

\begin{problem}
    In round $\cR_2$, we are allowed to re-allocate some students matched in round $\cR_1$ in order to match more students from $L$. Find a stability-preserving matching that maximizes the number of students who get matched from $L$, and subject to this, minimize the number of re-allocations made.

\end{problem}
    
\begin{problem}
   In the single round setting, given a set of students, and schools with strictly ordered preference lists $l(s),l(h)$ respectively, and a weight function $w(j)$ over the edges of students to schools, find a vector of capacities for the schools and a stable matching with respect to this vector that maximizes the total weight. 

\end{problem}

\begin{theorem}
Problems 1,2,3,4,5 are NP-hard.
\end{theorem}

\begin{proof}
We show NP-Hardness by reducing all the problems from the Cardinal Set-Cover problem which was shown to be NP-Complete by Karp \cite{karp1975}.
In the Set-Cover problem we are given a Universal Set $U = \{e_{1},...,e_{n} \}$ of elements, and a collection of Subsets, $S_i \subseteq U$. The goal is to find the smallest sub-collection of subsets that cover all elements.\\
The key idea in the reductions is to convert the subsets into schools, and to convert elements into students who want to go to the schools that they are elements of.

\begin{enumerate}
	\item For every set $S_i = \{ e_{i1},.., e_{ik} \}$ we construct a corresponding school $h_i$. The preference list for each $h_i$ is $(s_i', n_i, e_{i1},...., e_{ik})$, where $n_i \in N$ and $s_i'$ are students who only want to go to $h_i$. The order of $e_{ij}$ on $h_i's$ preference list does not matter, and the preference lists of $e_j$ can be arbitrary over the sets they are members of. We set the capacities of the schools to be 1 in round $\cR_1$, and all $s_i'$ are matched to $h_i$. In round $\cR_2$ the $n_i$ arrive. $max_L min_N$ will match all $e_j$, and the fewest possible $n_i$. By construction, if a student $e_j$ is admitted to $h_i$ then $n_i$ must also be admitted to $h_i$, otherwise $(n_i, h_i)$ will form a blocking pair. Therefore the admitted $n_i$ correspond to an optimal set cover.
	
	\item The reduction for $MAX_N MIN_L$ is symmetric to $MAX_L MIN_N$. The difference is that the gadgets $n_i$ are now students in $L$, and $e_i$ are now students in $N$.

	\item The reduction is the same as Problem 1. A solution to $k-max_L$ corresponds to the decision version of set cover (i.e. is there a cover of size $k$?). If all $L$ are matched then it is a yes-instance, otherwise it is a no-instance.
	
	\item We modify our reduction from Problem 1, by including a new school $h_0$ whose preference list is $(n_1,n_2,...)$. Each $n_i$ prefers $h_i$ to $h_0$. In round $\cR_1$, we set the capacity of $h_0$ to $n$ and the capacity of all the other schools to $1$. In round $\cR_1$, the $n_i$ are admitted to $h_0$. In round $\cR_2$, all $L$ will be admitted. If a student $e_j$ is matched to $h_i$ then $n_i$ must be reallocated to $h_i$, otherwise $(n_i, h_i)$ will form a blocking pair.
	An optimal solution will minimize the number of $n_i$ no longer matched to $h_0$; these $n_i$ correspond to an optimal set-cover.

	\item 
	We modify our reduction from Problem 4, by the inclusion of weights on the edges in the following way: $w(h_0, n_i)=1, w(h_j, n_i) = 0, w(h_j, e_{ij}) = 2$. If an $e_j$ is matched to $h_i$ then $n_i$ must also be matched to $h_i$. 
	An uncapacitated max-weight stable matching will match all the $e_j$, while minimizing the number of $(h_j, n_i)$ that are matched.

\end{enumerate}
\end{proof}

%% file: discussion.tex
\section{Discussion}
\label{sec.discussion}

The main remaining issue is obtaining mechanisms for Settings A1, A2 and B1 which are incentive compatible. One may also consider making the Settings more strict by allowing only complete lists: does this allow the design of incentive compatible mechanisms?

%% file: ack.tex
\section{Acknowledgements}
\label{sec.ack}

We wish to thank Laura Doval, Federico Echenique, Nicole Immorlica and Thorben Trobst for valuable discussions and pointers into the literature.